\documentclass[letterpaper, 10 pt, conference]{ieeeconf}  

\IEEEoverridecommandlockouts                              

\usepackage{amsthm}  
\usepackage{graphicx} 
\usepackage[utf8]{inputenc} 
\usepackage[T1]{fontenc}    
\usepackage{hyperref}       
\usepackage{url}            
\usepackage{booktabs}       
\usepackage{amsfonts}       
\usepackage{nicefrac}       
\usepackage[dvipsnames]{xcolor}
\usepackage{amsmath,amssymb}
\usepackage[american]{babel}

\usepackage{xspace}
\usepackage[linesnumbered,ruled,noend]{algorithm2e}
\usepackage{cite}
\usepackage{comment}
\usepackage{caption}
\usepackage{subcaption}
\usepackage{todonotes}
\usepackage{booktabs}
\usepackage{array}
\usepackage{amsmath}

\usepackage{eso-pic}
\newtheorem{proposition}{Proposition}

\newtheorem{definition}{Definition}

\newtheorem{example}{Example}

\usepackage[normalem]{ulem} 

\newcommand{\pX}{\mathbf{x}}
\newcommand{\pbigX}{\mathbf{X}}
\newcommand{\pY}{\mathbf{y}}

\newcommand{\pU}{\mathbf{u}}

\newcommand{\pbigU}{\mathbf{U}}

\newcommand{\pv}{\mathbf{v}}

\newcommand{\reals}{\mathbb{R}}
\newcommand{\naturals}{\mathbb{N}}

\SetKwFunction{MPC}{MPC}
\SetKwFunction{feasible}{feasible}
\SetKwFunction{Maxh}{Maxh}
\SetKwFunction{hybridStateEstimator}{hybridStateEstimator}

\newcommand{\safeset}{\mathcal{C}}

\DeclareMathOperator*{\argmin}{argmin}
\DeclareMathOperator*{\argmax}{argmax}

\title{\LARGE \bf
Feasibility-Guided Safety-Aware Model Predictive Control for Jump Markov Linear Systems  
}

\author{
Zakariya Laouar, Qi Heng Ho, Rayan Mazouz, Tyler Becker, Zachary N. Sunberg
\thanks{Authors are with the department of Aerospace Engineering Sciences at the University of Colorado Boulder, CO, USA {\tt\small \{zakariya.laouar@colorado.edu\}}.
}
}

\begin{document}

      
\maketitle
\thispagestyle{empty}
\pagestyle{empty}

\begin{abstract}
In this paper, we present a controller framework that synthesizes control policies for Jump Markov Linear Systems subject to stochastic mode switches and imperfect mode estimation. Our approach builds on safe and robust methods for Model Predictive Control (MPC), but in contrast to existing approaches that either optimize without regard to feasibility or utilize soft constraints that increase computational requirements, we employ a safe and robust control approach informed by the feasibility of the optimization problem. We formulate and encode finite horizon safety for multiple model systems in our MPC design using Control Barrier Functions (CBFs). When subject to inaccurate hybrid state estimation, our feasibility-guided MPC generates a control policy that is maximally robust to uncertainty in the system's modes. We evaluate our approach on an orbital rendezvous problem and a six degree-of-freedom hexacopter under several scenarios and benchmarks to demonstrate the utility of the framework. Results indicate that the proposed technique of maximizing the robustness horizon, and the use of CBFs for safety awareness, improve the overall safety and performance of MPC for Jump Markov Linear Systems.
\end{abstract}






\section{Introduction}



Advances in estimation and control techniques have enabled many applications of cyber-physical systems such as unmanned aerial systems (UAS), including military, industrial, and commercial  \cite{canis2015unmanned}. As these systems are increasingly deployed alongside humans in the real world, safety and reliability becomes increasingly important. The systems that operate in these environments must contend with the inevitable uncertainties present in the real world, which can be continuous or discrete in nature \cite{laurenti2020formal}. In this work, we consider Jump Markov Linear Systems (JMLS), a class of multiple model systems described by a discrete-time linear dynamic model whose parameters depend on discrete modes that the system switches between \cite{Blom1988}. Many real-world systems have imperfect models of the uncertain mode switches and imprecise sensors to detect these switches. This makes the problem of controlling these systems challenging due to the need to provide safety guarantees while accounting for the possibility of mode switches.  This paper focuses on this challenge and develops a principled framework for safety-aware control of JMLS subject to mode switches.


\begin{figure}[!t]
  \centering
    \includegraphics[width=0.7\linewidth]{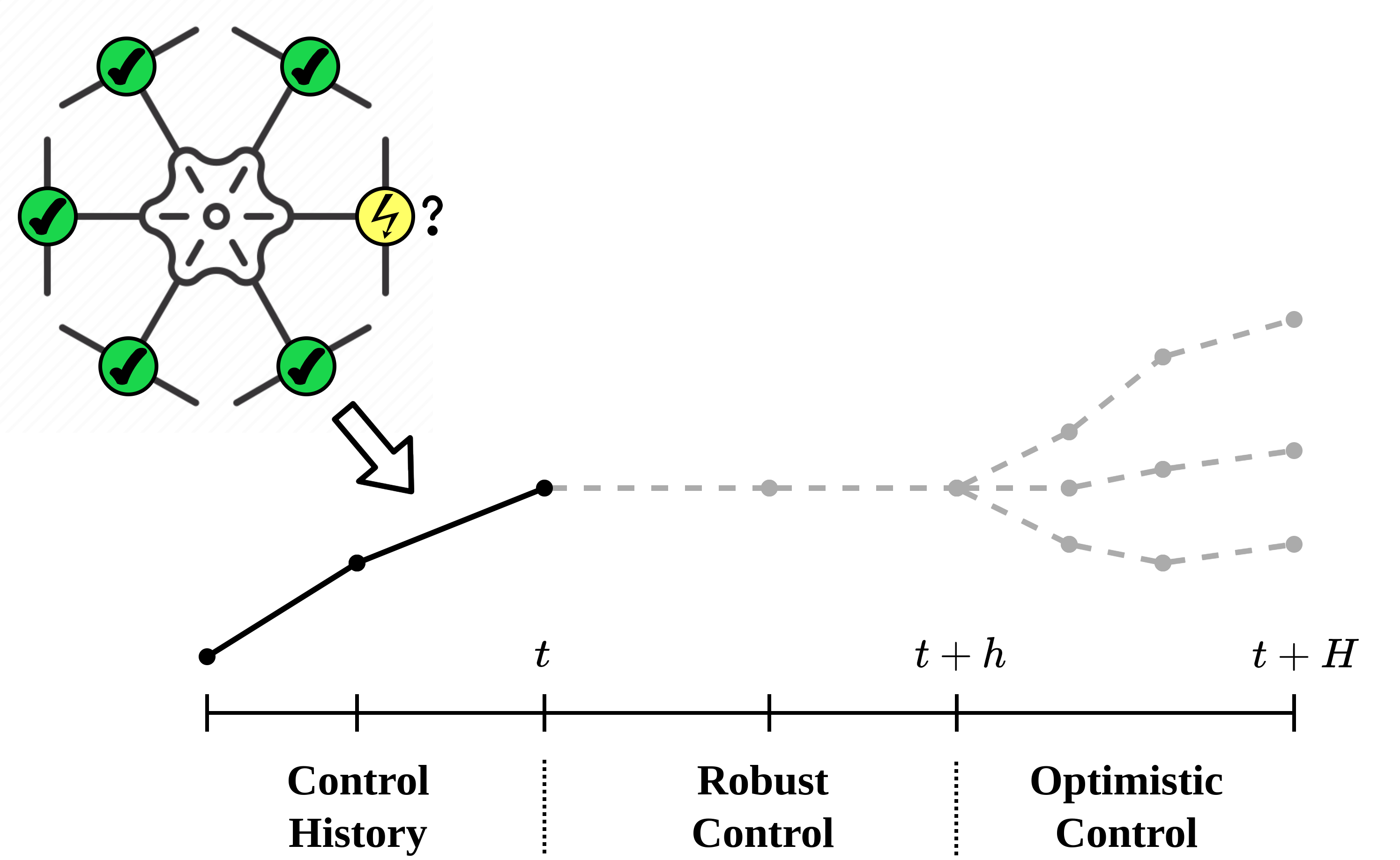}
    \caption{\small Feasibility-guided MPC for a multiple-model hexacopter subject to partially observable stochastic rotor faults. Our approach computes an adaptive robust control (consensus) horizon at each planning step to improve robustness and feasibility. 
    }
    \label{fig:robust}
\end{figure}

\begin{example}
    \label{ex:mineshaft inspection}
    Consider a mineshaft inspection scenario, where a hexacopter is tasked to inspect a narrow mineshaft while avoiding collisions with obstacles. The system is subject to partially observable stochastic rotor faults. At any time, any one of its six rotors may fail with low probability, causing its dynamics to change and the system to switch mode.
\end{example}
The hexacopter can be modeled as a JMLS subject to stochastic mode switches. Due to imperfect mode switching models and sensors, when a rotor has failed, the system may require some time to detect it. To ensure that the task is completed successfully and safely, a controller needs to ensure that the system is safe even if a mode switch occurs and is undetected. A typical method is to compute a solution robust to the worst case \cite{bemporad2007robust}. However, in many scenarios, optimizing for the worst case over the entire control horizon can lead to an infeasible optimization problem; There may not be a single control trajectory that is safe for all possible modes. Therefore, the system should determine the appropriate horizon of controller robustness, known as a \emph{consensus horizon}, by accounting for problem feasibility, as depicted in Fig. \ref{fig:robust}. In this work, we aim to design a controller that accounts for problem feasibility while being robust to stochastic mode jumps and mode switch detection delays.





\subsection{Related Work}
The general problem of planning under stochastic uncertainty can be modeled as a partially observable Markov decision process (POMDP) \cite{cassandra1998survey}. 
Unfortunately, solving a POMDP exactly is intractable due to the curses of dimensionality and history. Approximate solutions are often used in practice \cite{koch_dmu, sunberg2018online, lim2022generalized, gupta2022intention}, but do not perform well for large or continuous action spaces, which is a characteristic of multiple model systems.

We address safety-critical control for JMLS by proposing a Model Predictive Control (MPC) based approach. Scenario or multi-stage MPC approaches have been used to control systems under uncertainty, through sampling of scenarios \cite{micheli2022scenario}. These methods include a robustness horizon up to which scenarios are restricted to share the same action at each step in the planning horizon. However, a major limitation of scenario or particle approaches for controlling multiple model systems manifests in ensuring safety, since they have at most asymptotic safety guarantees. In our approach, we design a controller that maximizes robustness over mode uncertainty without sampling scenarios.
Several safety techniques have been proposed for MPC, including barrier approaches \cite{prajna2004safety, ames2014control, mazouz2022safety}, computation of reachable sets \cite{althoff2011reachable} and discrete approximations \cite{girard2006efficient}. Particularly, work \cite{cosner2023robust} leverages CBFs to generate policies that consider uncertainty in dynamics. In this paper, we utilize CBFs in MPC and extend it to the case of JMLS mode uncertainty.


Feasibility of MPC problems involving CBF constraints has been studied and enhanced in recent work \cite{zeng2021safety, zeng2021enhancing, ma2021feasibility, wabersich2022predictive}. \cite{zeng2021safety, zeng2021enhancing} enhance feasibility by reducing the CBF decay rate for single mode systems. \cite{ma2021feasibility} increase feasibility by using Generalized CBFs for inconsecutive timesteps. While this increases feasibility, safety is no longer enforced for intermediate timesteps. Another approach is using soft constraints via slack variables, compromising on safety guarantees \cite{wabersich2022predictive}. In this work, we discuss the impact of the size of the consensus horizon on the feasibility of the JMLS MPC problem with CBF constraints, and propose a method to improve feasibility while maximizing safety for each planning horizon.


\subsection{Contribution}

In this work, we present an MPC design that encodes the safety of JMLS with uncertain mode jumps and mode switch detection delays using CBFs. We provide a feasibility and safety analysis of the size of the consensus horizon, and propose a method to obtain the optimal consensus horizon while ensuring the finite-horizon optimization problem is feasible. We evaluate our control design on two case study scenarios. Results indicate that the proposed technique of maximizing the robustness horizon guided by the feasibility of the optimization problem and the use of CBFs are critical for the overall safety and performance of MPC for JMLS.



In summary, the main contributions of this paper are four-fold: (i) A safety-aware control design for controlling JMLS using CBF constraints, (ii) An analysis of the role of the consensus horizon on feasibility and safety of the closed-loop system, (iii) A practical method to obtain the optimal consensus horizon that maximizes robustness while ensuring problem feasibility at each time step, and (iv) Simulation case studies and benchmarks that demonstrate the safety and performance of the proposed approach.







\section{System Setup}
\label{sec:problem formulation}
In this work, we aim to develop safety-critical control policies for multiple-model systems subject to mode switches.
Consider a discrete-time Jump Markov Linear System
\begin{align}
    \label{sys:linear}
    \begin{split}
        \pX_{t+1} &= A_{i(t)}\pX_t + B_{i(t)} \pU_t,\\
        \pY_t &= H_{i(t)}\pX_t + \pv_t,
    \end{split}
\end{align}
\noindent where $t\in \naturals$ and system state $\pX_{t} \in \reals^n$. The control input $\pU_{t} \in \mathcal{U} \subset \reals^m $, where $\mathcal{U}$ is a compact set. The dynamics and control matrices of the $i(t)^{\text{th}}$ mode are $A_{i(t)} \in \reals^{n \times n}$ and $B_{i(t)} \in \reals^{n \times m}$, respectively. Measurement $\pY_{t}\in \reals^p$ is governed by $H_{i(t)} \in \reals^{p \times n}$, and the additive noise term $\pv_t\in \reals^p$.
The index $i(t)\in \{1, \dots, M\}$ represents the discrete \textit{mode} at a given time $t$, of which evolution is governed by a state and input independent finite state Markov chain%
%
\begin{equation}
\label{eq:markov}
    \mu_{t+1} = \Omega\mu_t,
\end{equation}
\noindent where $\Omega = \{\omega_{ij}\} \in [0,1]^{M\times M}$ defines the \textit{mode transition matrix} and $\mu_t\in [0,1]^M$ defines the \textit{categorical mode probability} at time $t$.


Due to noisy measurements, the system's state is estimated by a hybrid state estimator, such as an Interacting Multiple Model (IMM) filter \cite{Blom1988} which generates a mode probability distribution $\hat{\mu}_t$ and a mean continuous state $\hat{\pX}_t$ at each timestep $t$ from the measurements. Then, starting from an initial state $x_{0}$ and mode $i(0) \sim \mu_0$, under an output-feedback controller $\pi$, the evolution of the discrete-time closed-loop system is
\begin{align}
    \label{sys:clsystem}
    \begin{split}
        \pX_{t+1} &= A_{i(t)}\pX_t + B_{i(t)} \pi(\mathbf{Y}_t),\\
        \pY_t &= H_i\pX_t + \pv_t,\\
        i(t) &\sim \mu_t,\\
        \mu_{t+1} &= \Omega_{i(t)},
    \end{split}
\end{align}
where $\mathbf{Y}_t = \{\pY_0, \ldots, \pY_t\}$ and $\Omega_{i(t)}$ refers to the $i(t)^{\text{th}}$ column of the mode transition matrix.

In many safety-critical systems, the mode transition matrix $\Omega$ is an approximation \cite{Shalom2002}. As such, the state estimator may have inaccurate mode probability distributions, or require additional time to detect mode switches when they occur. In this work, we aim to develop a safety-aware controller $\pi$ robust to the inaccuracy of the mode estimate.



\section{Safety-aware Trajectory Planner Design}

The safety-aware trajectory planning problem consists of synthesizing a controller for a robot subject to uncertain mode switches. This section discusses our approach to designing a controller for this problem. The approach uses MPC with control barrier functions and a novel variable consensus horizon that balances between maximizing robustness over modes and feasibility of the optimization problem.

\subsection{Model Predictive Control}
To guarantee safety for the entire problem duration, a planner would have to reason about all possible sequences of states and modes across the planning horizon. This presents a significant computational challenge, rendering most problems computationally infeasible. However, open-loop planning can often provide a satisfactory approximation of an optimal closed-loop plan while greatly enhancing computational efficiency \cite{koch_dmu,erez2013realtime}. To this effect, we adopt the receding horizon paradigm of MPC.

We provide a base template MPC formulation for JMLS that we build upon in the following sections. At each time step, the Model Predictive Controller plans over a finite horizon $H$. We represent the stage cost as $q(\pX_{k}^i,\pU_k,\hat{\mu}_k)$ and the terminal cost as $p(\pX_{H}^i,\hat{\mu}_k)$, where $\pX_k^i$ represents the state for mode $i(k)$ at planning time $k$. The system is subject to the continuous dynamics in System \eqref{sys:linear} with mode-switching dynamics \eqref{eq:markov}, initial state $x_t$, and initial mode probability $\hat{\mu}_t$. The safety criteria, e.g. "avoid obstacle collisions" is represented by constraining the state for each mode $\pX_k^i$ to remain within a designer-specified set $\mathcal{S}_k$. In our approach, we utilize CBFs for $\mathcal{S}_k$. Finally, under some control constraints, the goal is to generate a control sequence that minimizes the cost function subject to the constraints. This yields the following optimal control problem (OCP):
\begin{subequations}
    \label{problem:safe}
    \begin{align}
        \pbigU_t^* := \argmin_{\pbigX_t, \pbigU_t} & 
        \quad p(\pX_{H}^{1:M},\hat{\mu}_k) + \sum_{k=0}^{H-1} q(\pX_{k}^{1:M},\pU_k^{1:M},\hat{\mu}_k)\\
        \text {s.t.} \quad & \forall k \in \{0, \dots, H-1\}, i \in \{1, \dots, M\},
        \notag\\
        & \pX^i_{k+1} = A_{i}\pX^i_k + B_{i} \pU_k^i, \label{optim:1eq1} \\ 
        & \pX^i_{k+1} \in \mathcal{S}_k \label{optim:safe},\\
        & \pX^i_0 = x_t, \\
        & \pU^i_k \in \mathcal{U}, \\
        & \hat{\mu}_{k} = \hat{\mu}_{t},  \label{optim:1eq5}
    \end{align}
\end{subequations}
where $\pbigX_t = \{\{\pX_k^i\}^M_{i=1}\}_{k=0}^{H-1}, \pbigU_t = \{\{\pU_k^i\}^M_{i=1}\}_{k=0}^{H-1}$ and the mode probability $\hat{\mu}$ is assumed to remain constant over $H$. We assume that $p$ and $q$ are quadratic cost functions, making \eqref{problem:safe} a Quadratic Program. 


At each instance of Problem \eqref{problem:safe}, a sequence of open-loop control inputs is generated, and only the first input is applied to System \eqref{sys:linear}. After each control input is applied, the system generates a new hybrid state observation. Then, the state is updated and Problem~\eqref{problem:safe} is solved again. This recursive process allows the closed-loop controller $\pi$ to overcome the inaccuracies of neglecting the acquisition of future state information due to open-loop planning. 

In this current formulation, each mode has a corresponding control input trajectory that may be unique with respect to the other trajectories. Thus, $\pi$ must somehow choose between these trajectories to execute the first control input. In Section~\ref{sec: mode consensus}, we discuss how to account for uncertainty in the modes and mode switches to make this decision.



\subsection{Control Barrier Functions}
CBFs provide a framework of safety by creating a link between the dynamical system \eqref{sys:linear} and a designer's expressed safety criteria \eqref{optim:safe}, guaranteeing the satisfaction of these criteria \cite{xiao2023control}. For JMLS, we leverage CBFs as state constraints of \eqref{optim:safe} to provide robustness against mode uncertainty. We first introduce our notion of safety for closed-loop JMLS using the concept of set invariance \cite{xiao2023control}.
\begin{definition}[Forward Invariance \& Safety]
    A compact set $\mathcal{C} \subset \reals^{n}$ is forward invariant for closed-loop System \eqref{sys:clsystem} if $\mathbf{x}_0 \in \mathcal{C}$ implies $\mathbf{x}_t \in \mathcal{C}$ for all $t \in \naturals$. In this case, we call System \eqref{sys:clsystem} \textit{safe} with respect to set $\mathcal{C}$. 
\end{definition}

\noindent Additionally, we say $\pi$ is \emph{safe} if the closed-loop system under $\pi$ is safe with respect to $\mathcal{C}$. Formally, we enforce the notion of safety using CBFs for discrete systems \cite{agrawal2017discrete}. 
To that extent, let set $\mathcal{C}$ be described by the 0-superlevel set of a continuously
differentiable function $\beta: \mathbb{R}^{n} \rightarrow \mathbb{R}$, such that
\begin{align}
\begin{split}
      \mathcal{C} & = \{ \mathbf{x} \in \reals^{n} \: | \:  \beta(\mathbf{x}) \geq 0 \}, \\
    \partial \mathcal{C} & =  \{ \mathbf{x} \in \reals^{n} \: | \:  \beta(\mathbf{x}) = 0 \}. \\
\end{split}
\end{align}
Function $\beta$ is a CBF if $\frac{\partial \beta}{\partial \mathbf{x}} \neq 0, \forall {x} \in \partial \mathcal{C}$, and there exists an extended class $\mathcal{K}_{\infty}$ function $\gamma$, such that for discrete-time System \eqref{sys:linear}, $\beta$ satisfies for a given time $k$
\begin{equation}
   \exists \, \mathbf{u}_k \in \mathcal{U} \text{ s.t. } \quad \Delta{\beta}(\pX_k,\pU_k)\geq-\gamma(\beta(\pX_k)),
   \label{eq:discretecbf}
\end{equation}
where $  \Delta{\beta}(\pX_k,\pU_k) := \beta(\pX_{k+1}) - \beta(\pX_k)$. 
For simplicity, we will use the scalar form of $\gamma$ throughout the manuscript: $0 < \gamma \leq 1$. From \eqref{eq:discretecbf} we see that $1 - \gamma$ defines a lower bound on the rate of exponential decrease of $\beta(\pX)$, i.e.,
$\beta(\pX_{k+1}) \geq (1-\gamma)\beta(\pX_k)$.

For our System described in \eqref{sys:linear}, the safe set is the set of states satisfying \eqref{eq:discretecbf} at each time $k$,
\begin{equation}
\label{eq:safe_set}
    \mathcal{S}_k = \{\pX_{t+k} \in \mathbb{R}^{n}: \beta(\pX_{t+k}) \geq (1-\gamma)\beta(\pX_{t+k-1})\}.
\end{equation}
Note, $\mathcal{S}_k$ is implicitly generated from the safe control input in \eqref{eq:discretecbf}. Since our problem is characterized by mode uncertainty due to imperfect modeling and estimation, we first represent safety with respect to each mode. When planning over $M$ modes, the mode-specific set is, $\forall i\in\{1,...,M\}$,
\begin{equation}
    \mathcal{S}_k^i = \{\pX_{t+k}^i \in \mathbb{R}^{n}: \beta(\pX_{t+k}^i) \geq (1-\gamma)\beta(\pX_{t+k-1}^i)\}.
    \label{eq:modesafeset}
\end{equation}
Eq. \eqref{optim:safe} requires the controller to be safe for each mode, separately. In the next section, we show how to enforce finite horizon safety across all modes using control consensus.

\subsection{Robustness via Control Consensus across Modes}
\label{sec: mode consensus}
The optimal solution to \eqref{problem:safe} contains the optimal trajectory for each mode $i$. When there is high uncertainty over modes, or when detection of mode switching is inaccurate, executing the control input for one mode optimistically, based on maximum likelihood, may lead to unsafe behavior if the mode is estimated incorrectly. A safer and more conservative approach to tackle mode uncertainty is by enforcing the control inputs of the modes to be the same across all modes for some number of time steps during planning \cite{dyro2021particle}. This number of steps is called the \emph{consensus horizon}.


\begin{definition}[Consensus Horizon]
    A consensus horizon $h \in \{0, \dots, H - 1\}$ is a horizon up to which the control input is constrained to be the same for all modes during planning, i.e., $\pU^i_k = \pU^j_k \quad \forall i,j \in \{1, \dots, M\}, \forall k\in\{0, \dots, h\}$.
\end{definition}

Thus, Problem \eqref{problem:safe} requires a constraint specifying that all modes share the same control input for $h$ planning time steps. This yields the updated formulation:
\begin{subequations}
    \label{problem:consensus}
    \begin{align}
       \pbigU_t^* := \argmin_{\pbigX_t, \pbigU_t} & 
        \quad p(\pX_{H}^{1:M},\hat{\mu}_k) + \sum_{k=0}^{H-1} q(\pX_{k}^{1:M},\pU^{1:M}_k,\hat{\mu}_k)\\
        \text {s.t.} \quad & \forall k \in \{0, \dots, H-1\}, i,j \in \{1, \dots, M\},
        \notag\\
        & \pX^i_{k+1} = A_{i}\pX^i_k + B_{i} \pU_k^i, \label{optim:9eq1} \\ 
        & \pX^i_{k+1} \in \mathcal{S}_k^i,\label{prob1:safe}\\
        & \pX^i_0 = x_t, \\
        & \pU^i_k \in \mathcal{U}, \\
        & \hat{\mu}_{k} = \hat{\mu}_{t},  \label{optim:9eq5} \\
        \quad & \forall l \in \{0, \dots, h-1\}, \notag \\
        & \pU^i_l = \pU^j_l \label{prob1:consensus constraints}.
    \end{align}
\end{subequations}



The interaction between mode-dependent CBF constraints \eqref{prob1:safe} and consensus horizon control constraints \eqref{prob1:consensus constraints} leads to robustness over all modes over the consensus horizon.




With \eqref{problem:consensus} now parametrized with a consensus horizon $h$, we are faced with the following question: \textit{what should $h$ be}?
The typical approach for multi-modal MPC is to use a user-defined fixed consensus horizon $h$ \cite{dyro2021particle}. One common approach is to set $h=1$, for one-step consensus, resembling hindsight optimization techniques. This approach reasons through a certainty equivalent lens that assumes the robot will resolve its uncertainty after the first timestep. Another common approach is to set $h=H$, for full consensus, corresponding to robust control methods. Conversely, such an approach plans under the assumption that the uncertainty will not be resolved throughout the planning horizon. Nonetheless, selecting $h$ is challenging due to the difficulty in estimating how quickly the mode can be determined during execution. Substantial trial and error has to be conducted to choose a sufficient consensus horizon a priori. Additionally, as we discuss in the following, a fixed a priori chosen consensus horizon may lead to safety and feasibility issues during closed-loop execution.

\subsection{Feasibility Analysis}
In this section, we discuss the potential safety and feasibility drawbacks of planning with a pre-defined consensus horizon. Specifically, we show that too large $h$ may lead to infeasible OCP while too small $h$ may have safety issues. First, note that a larger consensus horizon leads to a smaller feasible set. This is formalized in the following proposition.


\begin{proposition}
    \label{prop:consensus}
    Let $F_h$ be the feasible set for \eqref{problem:consensus} with a consensus horizon $h$. Given two consensus horizons $h_1 < h_2 \leq H$, $F_{h_1} \supseteq F_{h_2}$. 
    
\end{proposition}

\begin{proof}
    The feasible set of the optimization problem is the intersection of the feasible set for each independent constraint. A larger consensus horizon strictly adds constraints to the optimization problem. Therefore, the feasible set for $h_2$ cannot be a superset of the feasible set for $h_1$.
\end{proof}




Increasing the consensus horizon may shrink the feasible set and may lead to an infeasible optimization problem (if the feasible set becomes empty). Thus, we need to ensure that \eqref{problem:consensus} is parameterized with a \textit{feasible consensus horizon}.

\begin{definition}[Feasible Consensus Horizon]
    A feasible consensus horizon $h$ is a consensus horizon for which \eqref{problem:consensus} is feasible.
\end{definition}

On the other hand, a lower consensus horizon may be unsafe during closed-loop execution.


\begin{proposition}
    \label{theorem1}
    Let $\pi_A$ be a safe controller that generates a feasible sequence of control inputs that repeatedly solves \eqref{problem:consensus} with a consensus horizon $h_A = H$ for $t \in \{1,...,\infty\}$. Let controller $\pi_B$ repeatedly solve \eqref{problem:consensus} with a consensus horizon $h_B < H$. $\pi_B$ may be an unsafe controller.
\end{proposition}

\begin{proof}[Proof Sketch]
    Let $F_{h_{A,1}}$ and $F_{h_{B,1}}$ be the feasible set at $t=1$ for $\pi_A$ and $\pi_B$, respectively. From Proposition~\ref{prop:consensus}, $F_{h_{B,1}} \supseteq F_{h_{A,1}}$.

    Let $(\pbigX^*, \pbigU^*)$ be the optimal state and control inputs output by $\pi_B$ at $t = 1$. Since $F_{h_{B,1}} \supseteq F_{h_{A,1}}$, it is possible that $(\pbigX^*, \pbigU^*) \notin F_{h_{A,1}}$. That is, $(\pbigX^*, \pbigU^*)$ is not a feasible decision set for the horizon $h_{A,1}$. Therefore, when the first control input $\pU_1$ is executed for $\pi_B$, the system may reach a state $\pX_{B,2}$ that is not contained within the reachable set of states for $\pi_A$ at $t = 2$, i.e.
    \begin{align*}
        &\pX_{B,2} \in\\
        &\{\pX: \pX_2 =  A_{i}\pX_1 + B_{i} \pU_1, \pU_1 \in F_{h_{B,1}}, i \in \{1, \ldots, M\}\} 
        \\&\supseteq \{\pX: \pX_2 = A_{i}\pX_1 + B_{i} \pU_1, \pU_1 \in F_{h_{A,1}}, i \in \{1, \ldots, M\}\}.
    \end{align*}
    
    Then, at $t=2$, if $\pX_{B,2}$ is not in the reachable set of states for $\pi_A$, there exists scenarios in which the system starting from $\pX_{B,2}$ may be unsafe. In the worst case scenario, the optimization problem \eqref{problem:consensus} with $\pX_{B,2}$ may immediately be infeasible since the previous optimization at $t = 1$ did not account for the control consensus constraints for planning time steps $k = h_{B}$ to $k = h_{B} + 1$, leading to an empty feasible set for the optimization problem \eqref{problem:consensus} starting from $\pX_{B,2}$, whereas $\pX_{A, 2} \in \{\pX_2: \pX_2 = A_{i}\pX_1 + B_{i} \pU_1, \pU_1 \in F_{h_{A,1}}, i \in \{1, \ldots, M\}\}$ with $h_A = H$ is safe.
    
\end{proof}


Therefore, if a controller that maximizes the consensus horizon is safe, a controller operating with a lower consensus horizon is not necessarily safe. Propositions~\ref{prop:consensus} and ~\ref{theorem1} exemplify the dilemma in pre-determining a consensus horizon.

Consider Example~\ref{ex:mineshaft inspection} again. If the hexacopter adopts a longer consensus horizon and is close to an obstacle, there may be no single sequence of control inputs that allows it to avoid a collision with the mineshaft walls across all possible mode dynamics. In this case, the OCP becomes infeasible. However, reducing the consensus horizon may expand the feasible set and render the OCP feasible. On the other hand, when the problem is already feasible, we generally want the consensus horizon to be as large as possible to be more robust to mode uncertainty. A pre-defined fixed consensus horizon is unable to achieve such a behavior.

\subsection{Adaptive Consensus Horizon Planning}
\label{sec:meth}





Planning with a pre-defined fixed consensus horizon has drawbacks in feasibility (if consensus horizon is too large) and safety (if consensus horizon is too small). This implies that there exists a possibly varying optimal sequence of consensus horizons that is best able to balance between robustness and feasibility to guarantee the safety of the robot. However, since the controller uses a finite planning horizon in a receding horizon MPC, computing the optimal consensus horizon sequence beforehand is not generally possible; Designing a recursively feasible MPC controller is an open area of research \cite{Borrelli_Bemporad_Morari_2017}. Therefore, using the results of Proposition~\ref{prop:consensus} and~\ref{theorem1} as a heuristic, we propose a \emph{greedy} controller design. At each time step, we select a consensus horizon that is as large as possible while maintaining feasibility.






\begin{algorithm}[t!]
    \caption{Maximally Feasible Consensus \\
    Horizon Binary Search}
    \label{alg:binsearch}
    \SetAlgoLined
    \SetKwInOut{Input}{Input}\SetKwInOut{Output}{Output}
    \Input{$x,H$}
    \Output{$h^*$}
    $h_{max} \gets H-1$ \\
    $h_{min}, h_{best} \gets 0$ \\
    \While{$h_{min} < h_{max}$}
    {
        $h \gets \lfloor(h_{min} + h_{max})/2\rfloor$ \\
        \eIf{$\feasible(h)$}
        {
            $h_{min} \gets h + 1 $ \\
            $h_{best} \gets h$
        }
        {
            $h_{max} \gets h - 1$\\
        }
    }
    \eIf{$\feasible(h_{max})$}
    {
        $h \gets h_{max}$ \\
    }
    {
        $h \gets h_{best}$ \\
    }
    \Return $h^*=h$
    \rule{6cm}{0.6pt}\\
    \begin{itemize}
        \item[] \hspace{-10mm} where  
	    \item $\feasible(h)$: returns a boolean expressing the feasibility of \eqref{eq:maxh} with consensus horizon $h$.
	\end{itemize}
\end{algorithm}





In the first stage of each planning step, we compute the consensus horizon that maximizes $h$ such that the OCP is still feasible. We define such a consensus horizon as a \emph{maximally feasible consensus horizon}.



\begin{definition}[Maximally Feasible Consensus Horizon]
A maximally feasible consensus horizon $h^*$ is a feasible consensus horizon in which all horizons below $h^*$ are feasible and all horizons above $h^*$ are infeasible, i.e.,
\begin{align}
    \begin{split}
        \label{eq:maxh}
        h^* := \argmax_{{\pX^i_k}, {\pU^i_k}, h} & \quad h\\
        \text {s.t.} \quad & \text{\eqref{problem:consensus} is feasible.}
    \end{split}
\end{align}
\end{definition}


To compute $h^*$ efficiently, we first see that the feasibility of consensus horizons is an ordered monotonic sequence. That is, if $h_i$ is a feasible consensus horizon, $h_j < h_i$ is also a feasible consensus horizon for all $j \in \{0, \dots, i-1\}$. Using this insight, we propose to solve \eqref{eq:maxh} through a search over an ordered value\footnote{This search procedure bears similarities to the bisection method for quasiconvex optimization~\cite{boyd2004convex}.}. This can be efficiently done using, e.g., a binary search over the possible consensus horizon values $\{1,...,H-1\}$ as shown in Algorithm \ref{alg:binsearch}. For a given $h$, we construct an optimal control problem identical to \eqref{problem:consensus} with the objective function \eqref{optim:1eq1} replaced with an arbitrary constant. During each search iteration, we can tractably determine if a given $h$ is a feasible consensus horizon by solving the newly constructed problem parametrized with $h$ and obtaining a feasibility certificate. This reduces the optimization program from a Quadratic Program to an Linear Program. After computing $h^*$, we solve \eqref{problem:consensus} by setting the consensus horizon at that time step to $h=h^*$.

\subsection{Computational Complexity}
\label{sec:complexity}

In comparison to single-mode MPC OCP \cite{nan2022nonlinear}, each control input call from our approach is more computationally intensive in the following ways. First, a JMLS inherently admits a larger optimization problem, as there are more decision variables and constraints. The number of decision variables increases linearly in the number of modes, and the number of constraints increases linearly in the number of modes and consensus horizon. That is, a JMLS MPC OCP has $O(M\cdot T\cdot H \cdot m \cdot n)$ decision variables, compared to $O(T \cdot m \cdot n)$ of single-mode MPC OCP. Second, as compared to using a fixed consensus horizon, computing $h^*$ in Alg.~\ref{alg:binsearch} computes the feasibility of varying consensus horizon $\log_2 H$ number of times. Let $\textsc{T}$ be the maximum time taken to solve an OCP with a fixed $h$. Then, an upper bound on the time taken by our approach to solve \eqref{problem:consensus} is $(\lceil \log_2 H \rceil + 1)\textsc{T}$. However, each feasibility check is a Linear Program with a constant objective function, which can generally be solved faster than finding an optimal solution to the full OCP.

This added computation may lead to a slower solution output, which may lead to performance issues for on-board computation, depending on the size of the state space, number of system modes, and the solution time requirements of the system. If the ability to run in high frequency, e.g. real-time, is an important consideration, a practical solution is to start from $h=1$ and increment until $h^*$ is found or the planning time limit is reached. Another solution is to solve \eqref{problem:consensus} with $h=1$ as a backup controller in case a solution is not found in the planning time limit. Further, the search for $h^*$ can be parallelized, mitigating the added computation from maximizing the feasible consensus horizon.




\section{Simulation Evaluation}

In this section, we evaluate our approach on two robotic systems in simulation: a 6-dimensional spacecraft conducting orbital rendezvous and a 12-dimensional hexacopter system subject to rotor faults. We demonstrate the performance of our approach against several baselines. Simulations are implemented in Julia, and the quadratic programs are solved using the interior-point solver Clarabel.jl \cite{Clarabel}. The optimizations are all performed single-threaded on a computer with a nominally 3.7 GHz CPU. The computational performance of the algorithms is discussed at the end of this section.

For both domains, we represent the cost as a quadratic stage-wise function that includes a state reference error term $\|(\pX_{k}^i-\pX_{r})\|^2_Q$ and a penalty on the control effort $\|{\pU_k}\|^2_R$, where $Q, R \succ 0$:
\begin{align}
    \sum\limits_{i=0}^{M} \mu_k(i) \left[\sum_{k=0}^{H-1} (\|(\pX_{k}^i-\pX_{r})\|^2_Q + \|{\pU_k}\|^2_R)\right].
\end{align}
For simplicity, the terminal cost is set to zero. We weight the stage-wise cost per mode by its corresponding probability.

The purpose of our evaluation is to assess the ability of our approach to plan robustly with respect to mode uncertainty and estimator inaccuracy. To this end, at every time step, each system is given a categorical mode probability estimate and mean state estimate by an oracle filter. In each simulation, a mode switch is induced at a time unknown to the controller. Additionally, the oracle filter undergoes a specified delay, also unknown to the controller, in detecting when a mode switch occurs. We compare our approach vs. baselines across a range of mode-switch and estimation-delay times in addition to problem-specific parameters discussed below.

We compare our approach to the following baselines:
\begin{itemize}
    \item \textbf{First-Step Consensus}: Use consensus horizon $h = 1$. 
    \item \textbf{Full-Step Consensus}: Use consensus horizon $h = H$. 
    \item \textbf{Non-Robust}: No consensus constraints. Instead, plan for the maximum likelihood mode according to the categorical mode probability $\hat{\mu}$ given by the oracle filter,
\end{itemize}



\subsection{Case Studies}

\subsubsection{Spacecraft Orbital Rendezvous}
We first demonstrate our approach on a spacecraft rendezvous involving a chaser and target satellite. Following \cite{7798765}, the linearized system dynamics are characterized by the following Clohessy-Wiltshire-Hill equations: 
\begin{equation}
\label{eq:cwh}
    \Dot{\pX} = \begin{bmatrix}
                     & \mathbf{0_3} &  & & I_3 & \\
                    3n^2 & 0 & 0 & 0 & 2n & 0\\
                    0 & 0 & 0 & -2n & 0 & 0\\
                    0 & 0 & -n^2 & 0 & 0 & 0\\
                \end{bmatrix} \pX
                + \begin{bmatrix}
                    \mathbf{0_3} \\
                    I_3\\
                \end{bmatrix} \pU
\end{equation}

where the state $\pX=[x_r, y_r, z_r, \Dot{x_r}, \Dot{y_r}, \Dot{z_r}]$ represents the relative positions and velocities between the chaser and target satellites, the control input $\pU=[F_x, F_y, F_z]$ represents the thrust of the chaser satellite in the $x,y$ and $z$ directions. The parameter $n=\sqrt{\mu/a^3}$ is the mean motion of the target satellite, where $\mu$ is Earth's gravitational constant and $a$ is length of the target satellite orbit's semi-major axis.

The chaser satellite must maintain a safe relative distance to the target satellite defined the safe set $\safeset$: the intersection of the following halfspaces (in kilometers),

\begin{align}
    \begin{split}
        \beta_1(x) &= \{x: 6 \leq x_1 \leq -6\} \\
        \beta_2(x) &= \{x: 0 \leq x_2 \leq 4\} \\
        \beta_3(x) &= \{x: -10 \leq x_3 \leq 10\}
    \end{split}
\end{align}


We constructed the barriers with various values of $\gamma$ and found $\gamma = 0.05$ balances efficiency in reaching the objective with conservativeness with respect to the obstacles.

We model \eqref{eq:cwh} as a JMLS by constructing a 2-mode system parameterized by the target's nominal mean motion $n$ and off-nominal mean motion $n_p$. This models uncertainty in the chaser satellite's uncertainty over the mean motion of the target satellite due to, e.g., the target satellite operating under two possible orbit maneuver strategies. We compare our approach against the baselines across a range of off-nominal mean motion value, mode-switch times and delay times. We set the nominal mean motion $n=0.061$ and range the off-nominal mean motion $n_p$ from $0.041$ to $0.101$ with a step size of $0.01$. Additionally, we range the mode-switch times from $1$ to $40$ with a step size of $5$ and the delay times from $0$ to $5$ with a step size of $1$. Simulation parameters are shown in Tab.~\ref{table:params}.



\begin{table}[htbp]
\centering
\caption{Simulation Parameters}
\label{tab:simulation-parameters}
\scalebox{0.96}{
\begin{tabular}{@{}lll@{}}
\toprule
\textbf{Parameters} & \textbf{Orbital Rendezvous} &  \textbf{Mineshaft Inspection} \\
\midrule
Planning Horizon, $H$ & 300 s & 0.5 s \\
Sampling time-step, $\Delta t$ & 10 s & 0.05 s \\
Minimum control, $u_{\min}$& $-0.1$ N & 0.1 N \\
Maximum control, $u_{\max}$ & 0.1 N & 20 N \\
State weight matrix, $Q$ & {diag[$50_{3\times3}$, $10^{-1}_{3\times1}$]} & {diag[$50_{3\times3}$, $10^{-1}_{9\times1}$]} \\
Control weight matrix, $R$ & $0.01 \cdot \mathbf{I}_{3\times3}$ & $0.01 \cdot \mathbf{I}_{6\times6}$ \\
Starting position, $\mathbf{x_0}$ & $[0.01, 3.8, 0]^\top$ km & $[0, 0, 0]^\top$ m \\
Target location, $\mathbf{x}_\text{ref}$& $[1.0, 1.0, 0]^\top$ km & $[-0.7, 0.7, -5]^\top$ m \\
\bottomrule
\end{tabular}
}
\label{table:params}
\end{table}

\begin{table}[h]
\centering
\caption{Orbital Rendezvous simulation results of our approach against baselines. The best entries are bold.}
\begin{tabular}{|l|c|c|c|c|}
\hline
 & First-Step & Full-Step & Non-Robust & Ours \\ \hline
Total Trials & 336 & 336 & 336 & 336 \\ 
Successes & 188 & 230 & 65 & \textbf{336} \\ 
Success (\%) & 56.0 & 68.4 & 19.3 & \textbf{100.0} \\ 
Average Cost & 2.05 & 2.29 & 2.54 & \textbf{1.95} \\ 
\hline
\end{tabular}
\label{tab:system_planner_comparison_rendezvous}
\end{table}

\begin{figure}[!t]
  \centering
    \includegraphics[width=0.9\linewidth]{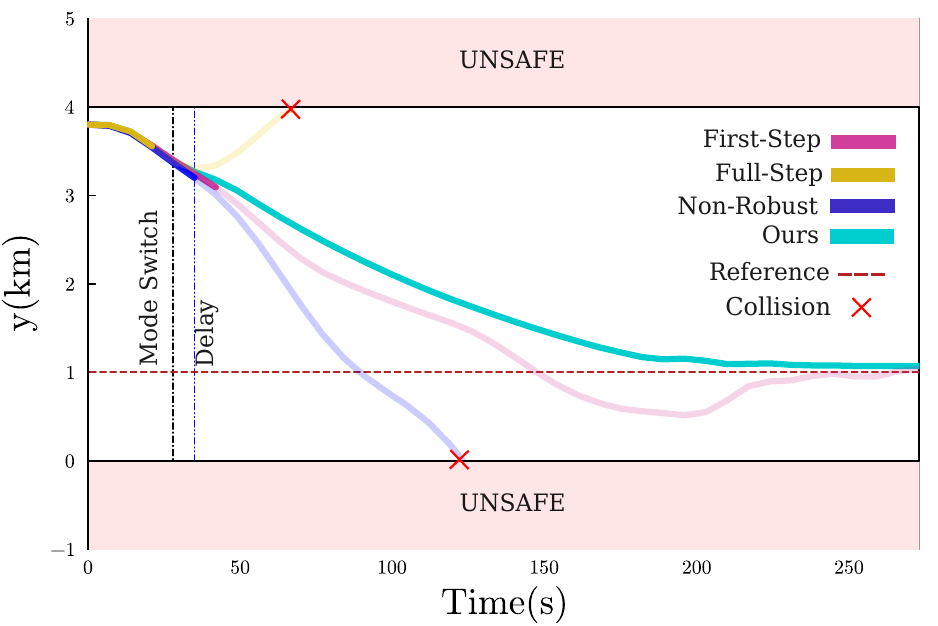}
    \caption{\small Spacecraft Orbital Rendezvous: example state trajectories for our approach and the baselines: First-Step Consensus, Full-Step Consensus, and Non-Robust. Only the in-track (y-direction) relative position is displayed to highlight the chaser spacecraft's (in)ability to stay behind the target satellite. The mode switch is induced at 35s and the estimation delay is 7s after. The simulation duration is 280s. Lower opacity illustrates problem infeasibility. When the problem is infeasible, the  first control input from the mode with the highest likelihood is chosen and executed.}
    \label{fig:rv_simulation}
    \vspace{-3mm}
\end{figure}

\subsubsection{Hazardous Mineshaft Inspection}
Next, we consider a safety-critical trajectory planning example whereby an autonomous hexacopter is tasked to inspect a hazardous mineshaft. The hexacopter must descend the mineshaft and reach a defined inspection point. At any time, the system can experience a complete rotor failure. Throughout its descent to the goal, the hexacopter must be safe with respect to the safe set $\safeset$, a three-dimensional rectangular polyhedron with no ceiling. Here $\safeset$ is defined as the intersection of the following halfspaces,

\begin{align}
    \begin{split}
        \beta_1(x) &= \{x: -1 \geq x_1 \leq 1\} \\
        \beta_2(x) &= \{x:  -1 \geq x_2 \leq 1\} \\
        \beta_3(x) &= \{x: x_3 \geq -6\}
    \end{split}
\end{align}

\noindent with $\gamma = 0.05$ for all barriers.



We consider a simplified representation of the hexacopter system, where two rotors are prone to failure at any time. Thus, the system has three modes, one for its nominal dynamics and two for the modes corresponding to the rotor failure. We compare our approach against the baselines across $9$ equidistant initial state conditions in the $xy$-plane representing the entrance of the mineshaft, over a range of mode-switch times (from $1$ to $40$) and delay times (from $0$ to $5$). Simulation parameters are shown in Tab.~\ref{table:params}.



\begin{table}[h!]
\centering
\caption{Mineshaft Inspection simulation results of our approach against baselines. The best entries are bold.}
\begin{tabular}{|l|c|c|c|c|}
\hline
 & First-Step & Full-Step & Non-Robust & Ours \\ \hline
Total Trials & 108 & 108 & 108 & 108 \\ 
Successes & 40 & 21 & 58 & \textbf{108} \\ 
Success (\%) & 37.0 & 19.4 & 53.7 & \textbf{100.0} \\ 
Average Cost & 5.81 & 4.94 & 4.74 & \textbf{3.92} \\ 
\hline
\end{tabular}
\label{tab:system_planner_comparison_hexacopter}
\end{table}

\begin{figure}[!t]
  \centering\includegraphics[width=0.95\linewidth]{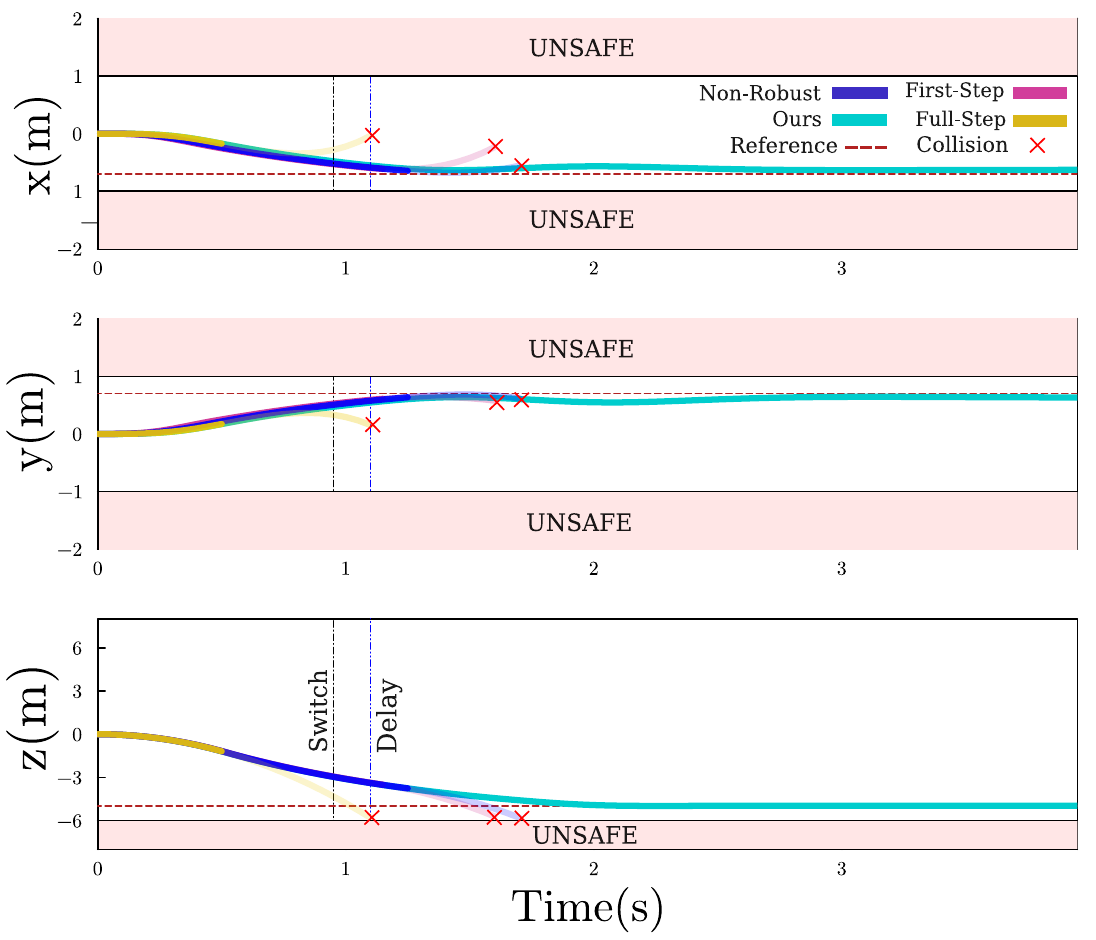}
    \caption{\small Hazardous Mineshaft Inspection: example state trajectories for our approach and the baselines: First-Step Consensus, Full-Step Consensus, and Non-Robust. The mode switch is induced at 0.9s and the estimation delay is 0.15s after. The simulation duration is 4 seconds. Lower opacity illustrates problem infeasibility. When the problem is infeasible, the first control input from the mode with the highest likelihood is chosen and executed.}
    \label{fig:hex_simulation}
    \vspace{-4mm}
\end{figure}

\subsection{Results Discussion}
Figs.~\ref{fig:rv_simulation} and ~\ref{fig:hex_simulation} depict illustrative example simulations of the compared approaches for the Spacecraft Orbital Rendezvous problem and Mineshaft Inspection problem, respectively. In Fig.~\ref{fig:rv_simulation}, the Full-Step consensus planner fails immediately after the mode switch is induced due to its inability to recourse after encountering infeasibility. Shortly after, the Non-Robust planner fails. The First-Step consensus planner also encounters infeasibility yet in this example, it manages to control the chaser satellite to the reference point, albeit at a higher average cost with respect to our planner. Meanwhile, our planner safely guides the spacecraft to the in-track (y-direction) reference point with the lowest average cost. In Fig.~\ref{fig:hex_simulation}, the Full-Step consensus planner fails first as well, this time followed by the First-Step consensus planner (1.1s) and finally the Non-Robust planner (1.25s). Meanwhile, our planner safely recovers the hexacopter after rotor failure and reaches the desired inspection reference point. Here, the discontinued $x$ and $y$ trajectories result from the colliding with the barrier in the $z$-direction.

Tab.~\ref{tab:system_planner_comparison_rendezvous} and Tab.~\ref{tab:system_planner_comparison_hexacopter} show the success rate and average cost of the compared planners for the Spacecraft Orbital Rendezvous problem and Mineshaft Inspection problem, respectively. The results show that our approach of adaptive consensus horizon planning is the safest and has the lowest cost across both problem environments. In the Spacecraft Orbital Rendezvous problem, the over-optimistic Non-Robust planner performs the worst with the highest average cost across all trials. The First-Step planner with $h = 1$ performs relatively less poorly, with lower cost and higher success rate. Nonetheless, its poor performance is attributed to its short consensus horizon causing feasibility issues when mode switch detections are delayed. The Full-Step planner with $h = H$ performs slightly better than both the First-Step and Non-Robust planners since the problem generally allows feasibility with a maximum consensus horizon. Our planner reaches the desired reference in $100\%$ of the trials, validating our approach of maximizing the consensus horizon subject to feasibility. In the Mineshaft Inspection problem, the Full-Step planner performs the worst out of all the planners in terms of success rate, while the First-Step planner obtains the highest cost. The Non-Robust planner performs better than the other baselines, but worse than our approach. In this problem, since the dynamics are significantly different between each mode and the safe sets are generally smaller, planning too optimistically (First-Step and Non-Robust) or too conservatively (Full-Step) leads to unsafe closed-loop behavior. Our approach is able to balance between robustness over mode uncertainty and problem feasibility, achieving a $100\%$ success rate with the lowest cost.

The results indicate that our approach of computing the maximally feasible consensus horizon allows more flexibility over different problem environments without needing to tune the consensus horizon for each new environment. The high success rate and low cost of our approach over both environments and varying parameters imply that the combination of the CBF constraints and the consensus horizon is critical for the overall safety and performance of MPC for JMLS.

Finally, Table~\ref{tab:algorithm-performance-data} shows the control input rate. Our controller optimizes for control inputs at an average rate of $9.89Hz$ for Orbital Rendezvous and $11.21Hz$ for Mineshaft Inspection. Our approach, although slower than the baselines, is able to maintain an operational control input computation rate while delivering a lower cost and higher success rate.

\begin{table}[h]
\centering
\caption{Control Input Computation Rate (Hz)}
\begin{tabular}{lcccc}
  \toprule
  & \multicolumn{2}{c}{\underline{Orbital Rendezvous}} & \multicolumn{2}{c}{\underline{Mineshaft Inspection}} \\
  Algorithm & Mean & Std. Dev. & Mean & Std. Dev. \\
  \midrule
  One-Step & 48.74 & 9.42 & 63.61 & 11.68 \\
  Full-Step & 44.37 & 10.78 & 67.08 & 15.29 \\
  Non-Robust & 53.97 & 6.51 & 66.16 & 10.88 \\
  Our Approach & 9.89 & 1.14 & 11.21 & 0.81 \\
  \bottomrule
\end{tabular}
\label{tab:algorithm-performance-data}
\end{table}

\section{Conclusion \& Future Work}

Our work presents an MPC framework that synthesizes safety-aware control policies for Jump Markov Linear Systems subject to stochastic mode switches. The MPC framework uses CBFs and tackles the robustness-feasibility tradeoff when asserting control consensus constraints over the planning horizon. Simulation experiments illustrate the utility of our proposed approach. For future work, we plan to extend the CBF constraints to include nonlinear safe sets and to account for noise in the dynamics. Additionally, we plan to explore techniques to address the added computation of computing the maximally feasible consensus horizon at each time step. Finally, we plan to validate our control design on hardware experiments. 



\section{Acknowledgments}
We thank Michael H. Lim for his insightful discussions.
This work was supported in part by NSF grant 203906.


\bibliographystyle{IEEEtran}
\bibliography{cite}

\end{document}